\documentclass[12pt,a4paper]{article}

\usepackage{amsmath}
\usepackage{amsfonts}
\usepackage{amssymb}
\usepackage{amsthm}

\newcommand{\Spin}{{\mathrm{Spin}}}

\newcommand{\sll}{{\mathfrak{sl}}}   
\newcommand{\SU}{{\mathrm{SU}}}
\newcommand{\U}{{\mathrm U}}

\newcommand{\so}{{\mathfrak{so}}}

\newcommand{\Cl}{{\mathrm{Cl}}} 

\newcommand{\R}{{\mathbb R}}   
\newcommand{\C}{{\mathbb C}}   
\newcommand{\HH}{{\mathbb H}}   
\newcommand{\Oct}{{\mathbb O}}   
\newcommand{\Z}{{\mathbb Z}}   

\DeclareMathOperator\End{End}

\theoremstyle{definition}
\newtheorem{theorem}{Theorem} 

\newtheorem{example}{Example}

\begin{document}

\title{Commuting Clifford actions}

\author{John W. Barrett
\\ \\
School of Mathematical Sciences\\
University of Nottingham\\
University Park\\
Nottingham NG7 2RD, UK\\
\\
E-mail john.barrett@nottingham.ac.uk}

\date{25th October 2024}

\maketitle

\begin{abstract} 
It shown that if a vector space carries commuting actions of two Clifford algebras, then the quadratic monomials using generators from either Clifford algebra determine a spinor representation of an orthogonal Lie algebra. 

Examples of this construction have applications to high energy physics, particularly to the standard model and unification. It is shown how to use Clifford data to construct spectral triples for the Pati-Salam model that admit an action of Spin(10).
\end{abstract}

\section{Introduction}
A Clifford module is a representation of a Clifford algebra on a complex vector space. A standard result is that taking quadratic monomials in the generators (gamma matrices) determines a spinor representation of the corresponding orthogonal Lie algebra. 

This paper studies  commuting actions of two real Clifford algebras,\\ $\Cl(p_1,q_1)$ and $\Cl(p_2,q_2)$, on a vector space $\mathcal H$. According to Theorems \ref{sorep} and \ref{spinrep}, the quadratic monomials in pairs of generators taken from either Clifford algebra determine a spinor representation of the Lie algebra $\so(q_1+p_2,p_1+q_2)$ on $\mathcal H$. This result is despite the fact that the gamma matrices from one Clifford algebra \emph{commute} with the gamma matrices of the other rather than anticommute, which would be the case if they both belonged to one larger Clifford algebra. Note that in some (but not all) cases the tensor product of two Clifford algebras can be made into one big Clifford algebra  \cite{h.blainelawsonSpinGeometry1990,barrettMatrixGeometriesFuzzy2015}, but this is \emph{not} the result reported here (although this does come into the proof of Theorem \ref{spinrep}).

The general results in this paper were inspired by several different examples (for specific $p_i$, $q_i$) that appear in \cite{fureyOneGenerationStandard2022} in connection with the fields of particle physics. In particular, one can take commuting actions of  $\Cl(0,q)$ and $\Cl(p,0)$ with any $p+q=10$, giving a representation of the unification group $\Spin(10)$ on a Hilbert space formed by the fermions of the standard model. 

For the case $q=6$ and $p=4$, the 
Clifford algebras are $\Cl(4,0)\cong M_2(\HH)$ and $\Cl(0,6)\cong M_8(\R)$, and the spin groups for each factor are $\Spin(4)$ and $\Spin(6)$.
Thus this structure determines, in a natural way,  a representation of the Pati-Salam group $G_{PS}=\Spin(4)\times\Spin(6)/\Z_2$ on $\mathcal H$ as a subgroup of $\Spin(10)$. The two subgroups lie in the even parts of these Clifford algebras, which are $\HH\oplus\HH$ and $M_4(\C)$. 

It is shown that a simple construction from this determines real spectral triples with Hilbert space $\mathcal H$ and algebra $\mathcal A=\HH\oplus\HH\oplus M_4(\C)$. The Clifford structures allow one to define an action of the unification group $\Spin(10)$ on the Hilbert space of these spectral triples.

\section{Spin group actions}

\subsection{Clifford modules}\label{sec:CliffordModules}
A real Clifford algebra \cite{chevalleyConstructionStudyCertain1955} is determined by a real vector space $V$ with a quadratic form $\eta$ (here assumed non-degenerate). A Clifford module \cite{atiyahCliffordModules1964} is a Clifford algebra together with a representation of the algebra in a complex vector space $\mathcal H$. Thus there is a map $c\colon V\to\End(\mathcal H)$, making $\mathcal H$ a left module for the algebra. Choosing a basis $\{e^a\}$ for $V$, the image of a basis vector is the `gamma matrix' $\gamma^a=c(e^a)$. The gamma matrices satisfy
\begin{equation}\gamma^a\gamma^b+\gamma^b\gamma^a=2\eta^{ab}.
\end{equation}
A Clifford module is said to be unitary if $\mathcal H$ is a Hilbert space and the operation $*$ of Hermitian conjugation in $\End(\mathcal H)$ is an involution on $V$. This makes the Clifford algebra a $*$-algebra.

The standard example of a Clifford algebra $\Cl(p,q)$ is for $V=\R^n$ with the quadratic form $\eta$ a diagonal matrix, with $p$ diagonal entries $+1$ and $q$ diagonal entries $-1$, with $n=p+q$. The unitary structure on a module for this Clifford algebra is chosen so that the gamma matrices $c(e^a)$, with $e^a$ the standard basis vectors, are unitary matrices (and are thus either Hermitian or anti-Hermitian).

The Clifford modules have standard definitions of a chirality operator $\gamma$ and an antilinear map $J$ called the real structure, with properties that depend on the signature of $\eta$, in fact on the parameter $s=q-p \mod 8$ \cite{barrettMatrixGeometriesFuzzy2015}. These properties are $J^2=\epsilon$, $J\gamma^a=\epsilon'\gamma^a J$ and $J\gamma=\epsilon''\gamma J$, with the signs given in Table \ref{signtable}.
\begin{table}
$$\begin{tabular}{|l|rrrrrrrr|}
  \hline
  s &0&1&2&3&4&5&6&7\\
\hline
$\epsilon$&1&1&-1&-1&-1&-1&1&1\\
$\epsilon'$&1&-1&1&1&1&-1&1&1\\
$\epsilon''$&1&1&-1&1&1&1&-1&1\\
  \hline
\end{tabular}$$
\caption{Signs for Clifford modules and real spectral triples.}\label{signtable}
\end{table}

It is worth explaining a bit more detail about these operators, as this will be useful later. The product of all of the gamma matrices in a Clifford module is denoted $P=\gamma^1\gamma^2\ldots\gamma^n$. A calculation shows that $P^2=(-1)^{s(s+1)/2}$, so the chirality operator is
\begin{equation}\gamma= i^{s(s+1)/2}P.\label{chirality}
\end{equation}
Obviously the definition of $P$ depends on the ordering of the basis, so $P$ and $\gamma$ are only determined by the Clifford module up to sign. 
An important point is that $P$ is in the Clifford algebra, so it is useful to express the structures in terms of $P$. In particular,
\begin{equation} JPJ^{-1}P^{-1}=\epsilon'.
\end{equation}

In the even $s$ cases, the table lists one real structure but there is a second one defined by
\begin{equation} \widehat J=JP
\end{equation}
This obeys $\widehat J^2=\epsilon''\epsilon$, $\widehat J\gamma^a=-\gamma^a \widehat J$ and $\widehat J\gamma=\epsilon''\widehat\gamma J$.

\subsection{Spin representations}\label{sec:SpinRep}
The quadratic monomials of $\Cl(p,q)$
\begin{equation}T^{ab}=\frac12\gamma^a\gamma^b
\end{equation}
with $a\ne b$, generate the Lie algebra $\so(p,q)$. Since $T^{ab}=-T^{ba}$, a spanning set is obtained by taking $a<b$. A calculation shows that the generators obey
\begin{equation}[T^{ab},T^{cd}]=\eta^{bc}\,T^{ad}-\eta^{ac}\,T^{bd}+\eta^{bd}\,T^{ca}-\eta^{ad}\,T^{cb}.
\end{equation}
Note that the generators $\widetilde T^{ab}=-T^{ab}$ obey the same relations but for the matrix $\widetilde \eta^{ab}=-\eta^{ab}$, which demonstrates the isomorphism $\so(p,q)\cong\so(q,p)$.

If the Clifford module is irreducible, the elements of $\mathcal H$ are called Dirac spinors (or pinors) and the representation of $\so(p,q)$ is called the Dirac spinor (or pinor) representation. 
\begin{table}
$$\begin{tabular}{|c|cccccccc|}
  \hline
  $s$ &0&1&2&3&4&5&6&7\\
\hline
\text{Structure}&$J,P$&$J$&$P$&$J$&$J,P$&$J$&$P$&$J$\\
\hline
\end{tabular}$$
\caption{Structure maps 
for the irreducible spinor representations of $\so(p,q)$.}\label{spinortable}
\end{table}

If $p+q$ is even (and greater than $0$), the Dirac spinor representation splits into two inequivalent representations of $\so(p,q)$ by the eigenspaces of the operator $P$. These representations are called Weyl spinor (or semispinor) representations, and are irreducible. 
Since $P$ is an invariant polynomial in the generators of $\so(p,q)$,
\begin{equation}P=\frac{2^{n/2}}{n!}\epsilon_{a_1\ldots a_n}T^{a_1a_2}T^{a_3a_4}\ldots T^{a_{n-1}a_n},\label{Ppolynomial}
\end{equation}
 it is a Casimir operator and so its eigenvalue characterises the representation. 

If $p+q$ is odd, the Dirac spinor is already an irreducible representation of $\so(p,q)$. Although in the odd case there are two inequivalent representations of the Clifford algebra, distinguished by the eigenvalues $\pm1$ of the chirality operator, the corresponding representations of $\so(p,q)$ are equivalent.

Each irreducible representation of $\so(p,q)$ has structure maps $P$ and/or $J$ commuting with the Lie algebra action. These are shown in Table \ref{spinortable}.  Note that the table is symmetrical if one replaces $s$ with $-s$. For $s=2,6$, the eigenvalues of $P$ are $\pm i$ and so the map $J$ does not survive the projection to the Weyl spinors. However, $J$ relates the two Weyl spinor representations as complex conjugates of each other.

\subsection{Commuting actions}\label{sec:commuting}

Consider a finite-dimensional vector space $\mathcal H$ with commuting actions of two Clifford algebras. The following result is that this determines a representation of a spin group using the generators from both algebras. 

\begin{theorem} \label{sorep} Let $\mathcal H$ be a Clifford module for both $\Cl(p_1,q_1)$ and $\Cl(p_2,q_2)$, such that the two actions commute. Denote the gamma matrices for the first action by $\{\Gamma_1^a\}$ and for the second action $\{\Gamma_2^\alpha\}$.
 The quadratic monomials formed from pairs of different generators in the set $\{\Gamma_1^a,\Gamma_2^\alpha\}$ determine a representation of the Lie algebra $\so(q_1+p_2,p_1+q_2)$ on $\mathcal H$.
\end{theorem}

\begin{proof} Define generators of the Lie algebra by $T_1^{ab}=\frac12\Gamma_1^a\Gamma_1^b$, $T_2^{\alpha\beta}= \frac12\Gamma_2^\alpha\Gamma_2^\beta$ (for $a\ne b$ and $\alpha\ne\beta$) and $U^{a\beta}=\frac12\Gamma_1^a\Gamma_2^\beta$ for all $a,\beta$. For convenience, define $T_1^{aa}=T_2^{\alpha\alpha}=0$. Then the Lie brackets are
\begin{equation}
\begin{aligned}\relax
[T_1^{ab},T_1^{cd}]&=\eta_1^{bc}\,T_1^{ad}-\eta_1^{ac}\,T_1^{bd}+\eta_1^{bd}\,T_1^{ca}-\eta_1^{ad}\,T_1^{cb}\\
[T_2^{\alpha\beta},T_2^{\gamma\delta}]&=\eta_2^{\beta\gamma}\,T_2^{\alpha\delta}-\eta_2^{\alpha\gamma}\,T_2^{\beta\delta}+\eta_2^{\beta\delta}\,T_2^{\gamma\alpha}-\eta_2^{\alpha\delta}\,T_2^{\gamma\beta}\\
[T_1^{ab},U^{c\delta}]&=\eta_1^{bc}\,U^{a\delta}-\eta_1^{ac}\,U^{b\delta}\\
[U^{a\beta},T_2^{\gamma\delta}]&=\eta_2^{\beta\gamma}\,U^{a\delta}-\eta_2^{\beta\delta}\,U^{a\gamma}\\
[U^{a\beta},U^{c\delta}]&=\eta_1^{ac}\,T_2^{\beta\delta}-\eta_2^{\beta\delta}\,T_1^{ca}
\end{aligned}
\end{equation}
Note that these are not the generators of $\so(p_1+p_2,q_1+q_2)$ because of the signs in the last line, which are due to the fact that $\Gamma_1^a$ and $\Gamma_2^\alpha$ commute. Instead, the generators are mapped to those of $\so(q_1+p_2,p_1+q_2)$ by defining $n_1=p_1+q_1$ and
\begin{equation}
\begin{aligned}
&T^{ab}=-T_1^{ab}\\
&T^{\alpha+n_1,\beta+n_1}=T_2^{\alpha\beta}\\
&T^{a,\beta+n_1}=U^{a\beta}\\
&T^{\alpha+n_1,b}=-U^{b\alpha}\\
\end{aligned}\label{productrep}
\end{equation}
and $\eta=(-\eta_1)\oplus \eta_2$.
\end{proof}
Note that by exchanging the two actions, one can also describe the Lie algebra as $\so(q_2+p_1,p_2+q_1)$, which is isomorphic, as noted in Section \ref{sec:SpinRep}. The relative signature flip is a bit surprising, and one might wonder what happens with the commuting actions of three Clifford algebras. However in that case, the commutators of quadratic expressions do not close to a Lie algebra.

The representation determined by Theorem \ref{sorep} is determined by spinor representations of the Lie algebra, as one might expect. The Hilbert space splits into irreducible representations of the two Clifford algebras, so it suffices to consider this case.

\begin{theorem} \label{spinrep} Suppose that there are irreducible representations of $\Cl(p_1,q_1)$ on $\C^{k_1}$ with gamma matrices $\gamma_1^a$, and $\Cl(p_2,q_2)$ on $\C^{k_2}$ with gamma matrices $\gamma_2^\alpha$. Then the representation of the Lie algebra $\so(q_1+p_2,p_1+q_2)$ on $\C^{k_1}\otimes\C^{k_2}$ given by Theorem \ref{sorep} with $\Gamma_1^a=\gamma_1^a\otimes 1$ and $\Gamma_2^\alpha=1\otimes \gamma_2^\alpha$ is equivalent to a Dirac spinor representation if $n_1=p_1+q_1$ is even or $n_2=p_2+q_2$ is even. If both $n_1$ and $n_2$ are odd, the representation is one of the two Weyl spinor representations.
\end{theorem}
\begin{proof} Consider the odd-odd case. According to \cite{barrettMatrixGeometriesFuzzy2015}, an irreducible Clifford module of type $(p,q)=(q_1+p_2,p_1+q_2)$ is given by the gamma matrices
\begin{equation}\begin{pmatrix}0&\gamma_1^a\\-\gamma_1^a&0\end{pmatrix}\otimes 1, \quad\begin{pmatrix}0&1\\1&0\end{pmatrix}\otimes \gamma_2^\alpha
\end{equation}
Note that the first set of matrices square to $-\eta_1^{aa}$, so the signature for these is reversed.
The quadratic monomials for this Clifford module are
\begin{equation}\frac12\begin{pmatrix}-\gamma_1^a\gamma_1^b&0\\0&-\gamma_1^a\gamma_1^b\end{pmatrix}\otimes 1, \quad
\frac12\begin{pmatrix}\gamma_1^a&0\\0&-\gamma_1^a\end{pmatrix}\otimes \gamma_2^\beta, \quad
\frac12\begin{pmatrix}1&0\\0&1\end{pmatrix}\otimes \gamma_2^\alpha\gamma_2^\beta.
\end{equation}
generating the Lie algebra $\so(q_1+p_2,p_1+q_2)$.
 The matrix $t=\begin{pmatrix}1&0\\0&-1\end{pmatrix}\otimes 1$ commutes with these generators, and so the $\pm$  eigenspaces of $t$ are the two different spinor representations. However, the $t=1$ eigenspace is exactly the representation \eqref{productrep}, which proves the result.
 
 Now suppose $n_1$ is even. Then an irreducible Clifford module of type $(q_1+p_2,p_1+q_2)$ is given by the gamma matrices
\begin{equation}
i\gamma_1^a\otimes 1,\quad \gamma_1\otimes \gamma^\beta_2\label{evengamma}
\end{equation}
using the chirality operator $\gamma_1$, which anti-commutes with the $\gamma_1^a$.
Therefore the quadratic expressions formed from these,
\begin{equation}S^{ab}=-\frac12\gamma_1^a\gamma_1^b\otimes 1, 
\quad S^{a, n_1+\beta}= \frac i2\gamma_1^a\gamma_1\otimes\gamma^\beta_2,
\quad S^{n_1+\alpha, n_1+\beta}=\frac12 1\otimes \gamma^\alpha_2\gamma^\beta_2
\end{equation}
generate the Dirac spinor representation of $\so(q_1+p_2,p_1+q_2)$. Define the operator $V=\exp(i\pi\gamma_1/4)\otimes 1$ (which is unitary if the Clifford module is unitary). This has the property that $V(\gamma^a_1\otimes 1)V^{-1}=i\gamma_1\gamma^a_1\otimes 1$. This operator determines the equivalent representation
\begin{equation}
\begin{aligned}
VS^{ab}V^{-1}&=-\frac12\gamma_1^a\gamma_1^b\otimes 1=-T_1^{ab}\\
VS^{a,\beta+n_1}V^{-1}&=\frac12\gamma_1^a\otimes\gamma_2^\beta=U^{a\beta}\\
VS^{\alpha+n_1,\beta+n_1}V^{-1}&=\frac12 1\otimes \gamma^\alpha_2\gamma^\beta_2=T_2^{\alpha\beta}\\
\end{aligned}
\end{equation}
which agrees exactly with \eqref{productrep}. There is a similar argument if $n_2$ is even.
\end{proof}

It remains to construct the structure maps for the spin representations constructed in Theorem \ref{spinrep}.
It is useful to write $p=q_1+p_2$ and $q=p_1+q_2$. Then, $s=q-p=s_2-s_1$, with $s_i=q_i-p_i$. For the cases when $s$ is even, the operator $P$ can be defined so that its eigenvalues distinguish the Weyl spinor submodules in the results of Theorem \ref{spinrep}.
In the even-even case, multiplying the gamma matrices in \eqref{evengamma} gives
\begin{equation}P=(-1)^{n_1/2}P_1\otimes P_2,
\end{equation}
which is unchanged by the unitary transformation, i.e., $VP V^{-1}=P$. In the odd-odd case, multiplying the gamma matrices shows that 
\begin{equation}P=(-1)^{(n_1-1)/2} P_1\otimes P_2.
\end{equation}

An antilinear structure map $J$ for the tensor product is obtained from a pair of structure maps, one for each factor, providing they either both commute with the gamma matrices or both anti-commute with the gamma matrices. The formula
\begin{equation} J=J_1\otimes J_2\label{J}
\end{equation}
commutes with the Lie algebra in the even-even cases, the odd-odd cases where $s=0$ or $4$, or the even-odd (and odd-even) cases
where $s_1$ (or $s_2$) is $3$ or $7$.
In the even-odd cases  with $s_2=1$ or $5$ one has to take 
\begin{equation} J=\widehat J_1\otimes J_2
\end{equation}
and similarly in the odd-even cases with $s_1=1$ or $5$,
\begin{equation} J= J_1\otimes \widehat J_2.
\end{equation}
The remaining odd-odd cases don't have a general formula for $J$ because the Weyl spinor representation does not have an antilinear structure.
In the even-even cases there is also the antilinear structure map
\begin{equation} \widehat J=JP=(-1)^{n_1/2}\widehat J_1\otimes \widehat J_2.\label{Jhat}
\end{equation}

\section{Examples}
Instances of this construction in particle physics appear in \cite{fureyOneGenerationStandard2022}, including Examples \ref{ex:1} and \ref{ex:2}.

\begin{example}\label{ex:1} The quaternion algebra $\HH$ has an irreducible representation $\sigma\colon\HH\to M_2(\C)$. Applied to the 
 imaginary quaternions $q_1,q_2,q_3\in\HH$ this gives the gamma matrices $\gamma_1^a=\sigma(q_a)$ for a Clifford algebra  $\Cl(0,3)$ acting in $\C^2$. The number $\gamma_2^1=i$ acting in $\C$ is the single gamma matrix for the type $(0,1)$ Clifford algebra. The tensor product of the two underlying real algebras is $\HH\otimes_\R\C_\R=M_2(\C_\R)$, which acts in $\C^2\otimes\C\cong\C^2$ by matrix multiplication. The generators of $\so(3,1)$ are $T_1=\{\frac12\sigma(q_1),\frac12\sigma(q_2),\frac12\sigma(q_3)\}$ and $U=\{\frac i2\sigma(q_1),\frac i2\sigma(q_2),\frac i2\sigma(q_3)\}$ and the representation in $\C^2$ is a Weyl spinor representation. This agrees with the holomorphic representation of $\sll(2,\C)$.
\end{example}

\begin{example}\label{ex:2} The quaternions $\HH$ and octonions $\Oct$ can be considered real vector spaces. Let $\mathcal H=\C\otimes_\R\HH\otimes_\R\Oct$, a complex vector space by scalar multiplication in the first factor. This has the structure of a non-associative algebra (the Dixon algebra). Left multiplication by the imaginary octonions generates a module for $\Cl(0,7)$ and left multiplication by $i$ times the imaginary quaternions generates a module for $\Cl(3,0)$. 
Applying Theorem \ref{spinrep} gives a  representation of $\Spin(10)$ on $\mathcal H\cong\C^{16}\otimes\C^2$ that is a Weyl spinor representation with multiplicity two.
\end{example}

The next example is new.
\begin{example}  \label{ex:PS}
Consider the tensor product of unitary irreducible modules for $\Cl(4,0)$ and $\Cl(0,6)$, as in Theorem \ref{spinrep}. This results in the Dirac spinor representation of $\Spin(10)$ on $\mathcal H=\C^4\otimes\C^8=\C^{32}$.
\end{example}
This construction has a distinguished subgroup of $\Spin(10)$, the Pati-Salam group 
$G_{PS}=\Spin(4)\times\Spin(6)/\Z_2$. It is shown in the next section that Example \ref{ex:PS} can be used to construct spectral triples for the internal space of the Pati-Salam model. 

\section{Spectral triple}
A real spectral triple \cite{connesNoncommutativeGeometryReality1995} is a mathematical structure for a Dirac operator on a spin manifold or its non-commutative generalisation. There is a $*$-algebra $\mathcal A$ generalising the algebra of functions in the manifold case. The fermion fields form a Hilbert space $\mathcal H$ that is a bimodule over $\mathcal A$.
The real structure is an antilinear operator $J$ that interchanges the left and right actions, generalising the real structure of a Clifford module. There is a Dirac operator $D$ and a chirality operator $\gamma$ on $\mathcal H$ obeying the same relations with $J$ as in Section \ref{sec:CliffordModules}, with $D$ in place of $\gamma^a$.

 An example of a spectral triple is constructed here, starting from the data in Example \ref{ex:PS}.
 The real spectral triple is defined using the same Hilbert space $\mathcal H, J=J_1\otimes J_2, \gamma=\gamma_1\otimes\gamma_2$ and algebra
\begin{equation}\mathcal A=\Cl(4,0)_\text{even}\oplus\Cl(0,6)_\text{even}\cong (\HH\oplus\HH)\oplus M_4(\C).
\end{equation}
Denote the projection onto the positive chirality eigenspace for the second action by $\pi_2^+$, and write $\pi_2^-=1-\pi_2^+$. Then $J\pi_2^+=\pi_2^-J$.

The left action of $a=(a_1,a_2)\in\mathcal A$ on $\mathcal H$ is defined to be
\begin{equation}l(a)=c_1(a_1)\otimes \pi_2^+ + 1\otimes c_2(a_2)\pi_2^-.
\end{equation}
The right action of the spectral triple can be computed from it
\begin{equation}r(a)= Jl(a^*)J^{-1}=c_1(a_1^*)\otimes\pi_2^- + 1\otimes c_2(a_2^*)\pi_2^+
\end{equation}
and it can be seen that the \emph{zeroth-order condition}
\begin{equation}[l(a),r(b)]=0
\end{equation}
is satisfied, making $\mathcal H$ a bimodule over $\mathcal A$. This part of the construction gives the same Hilbert space and algebra as the Connes-Chamsesddine spectral triples for the internal space of the Pati-Salam model \cite{chamseddineConceptualExplanationAlgebra2007, chamseddineSpectralStandardModel2013}. An internal space for a particle physics model has the same fields but with the spacetime manifold replaced by a single point.

In noncommutative geometry, the gauge group consists of unitary elements of the algebra ($u^*=u^{-1}$) acting on $\mathcal H$ by the adjoint action
\begin{equation} u\mapsto l(u)r(u^*).
\end{equation}
However, if one takes \emph{all} unitary elements, the resulting group in this example is $\SU(2)\times\SU(2)\times\U(4)/\Z_2$, which is larger than $G_{PS}$. One has to add a rather ad hoc extra condition called Connes' unimodular condition \cite{connesNoncommutativeGeometryStandard2006} to reduce the gauge group to $G_{PS}$. This condition is that $\det(l(u))=1$. 

However in the current construction, the Clifford structure gives a natural alternative way of defining the gauge group. It is defined by the elements  $u=(u_1,u_2)\in\Spin(4)\times\Spin(6)$, which form a subgroup of the unitary elements of $\mathcal A$. This subgroup automatically satisfies the unimodular condition. The gauge group action on $\mathcal H$ is the adjoint action
\begin{multline}l(u)r(u^*)=(c_1(u_1)\otimes \pi_2^+ + 1\otimes c_2(u_2)\pi_2^-)(c_1(u_1)\otimes\pi_2^- + 1\otimes c_2(u_2)\pi_2^+)\\=c_1(u_1)\otimes c_2(u_2),
\label{gauge}\end{multline}
which gives a faithful action of the correct subgroup $G_{PS}\subset\Spin(10)$. Another bonus from this construction is that it automatically leads to the correct charges of the fermions in $\mathcal H$ for one generation of the Standard Model in a conceptual way: it is because the representation of $G_{PS}$ extends to a Dirac spinor representation of $\Spin(10)$ by an application of Theorem \ref{spinrep}. This is the correct representation for the gauge charges of the fermion fields and their conjugates \cite{baezAlgebraGrandUnified2010}.

Dirac operators for this spectral triple are defined by
\begin{equation}D=c_1(d)\otimes 1\label{PSDirac}
\end{equation}
where $d$ is any Hermitian element of $\Cl(4,0)_\text{odd}$. This can be written in terms of the gamma matrices as $D=d_a\gamma^a_1\otimes 1$ for a vector $(d_1,d_2,d_3,d_4)\in\R^4$.   The operator $D$ is Hermitian, commutes with $J$, anticommutes with $\gamma$ and satisfies the \emph{first order condition}
 \begin{equation}[[D,l(a)],r(b)]=0
\end{equation}
for all $a,b\in\mathcal A$, defining a real spectral triple of $KO$-dimension $s=6$. 

In an internal space, the physical interpretation of the components of the Dirac operator depends on the action of the gauge group,
\begin{equation}D\mapsto l(u)r(u^*) \,D\, l(u^*)r(u).
\end{equation}
The components of $D$ that are invariant are simply constants, and so determine mass parameters in the Dirac action, whereas the non-invariant components are Higgs fields having a Yukawa interaction with the fermions.

Here, due to \eqref{gauge}, the gauge transformation is $d\mapsto u_1du_1^*$. Therefore $d$ is the electro-weak Higgs field, with the gauge group acting in the vector representation of $\Spin(4)$. Note that a Clifford algebra interpretation of this Higgs field has been given previously in the context of the graded tensor product of Clifford algebras \cite{todorovSuperselectionWeakHypercharge2021}.


As a final remark, one can easily generalise the construction of spectral triples to other Clifford algebras, but it is not yet clear which ones are interesting for physical models.

\section*{Acknowledgement} Thanks are due for the hospitality of Nichol Furey and the Physics Department of Humboldt University and a visitor grant from the Kolleg Mathematik Physik Berlin that enabled the initial phase of this research.

\bibliographystyle{plain}

\bibliography{../../JohnZoteroLibrary}

\end{document}